\documentclass[12pt]{article}

\usepackage{fullpage}
\usepackage{amsmath,amsthm,amssymb}
\usepackage{mathrsfs}
\usepackage{hyperref}
\usepackage{longtable}

\newcommand{\REF}{\mathrm{REF}}
\newcommand{\RREF}{\mathrm{RREF}}

\newcommand{\Dom}{\mathrm{Dom}}
\newcommand{\Img}{\mathrm{Img}}

\newcommand{\Ref}{\textup{-REF}}

\newcommand{\res}{{\upharpoonleft}}

\newcommand{\nats}{\mathbb{N}}

\newcommand{\N}{\nats}
\newcommand{\Vars}{\mathrm{V}}
\newcommand{\Bool}{\mathrm{B}}

\newcommand{\PTIME}{\mathrm{P}}
\newcommand{\NPTIME}{\mathrm{NP}}

\newcommand{\THREESAT}{\mathrm{3\text{-}SAT}}
\newcommand{\Ppoly}{\mathrm{P/poly}}

\newcommand{\QPTIME}{\textup{QP}}
\newcommand{\SUBEXP}{\textup{SUBEXP}}
\newcommand{\WP}{\textup{W[P]}}
\newcommand{\FPT}{\textup{FPT}}

\theoremstyle{theorem}
\newtheorem{theorem}{Theorem}
\newtheorem{lemma}[theorem]{Lemma}

\newtheorem{claim}[theorem]{Claim}

\theoremstyle{definition}
\newtheorem{example}[theorem]{Example}
\newtheorem{remark}[theorem]{Remark}

\renewcommand{\le}{\leqslant}
\renewcommand{\leq}{\leqslant}
\renewcommand{\ge}{\geqslant}
\renewcommand{\geq}{\geqslant}

\begin{document}

\title{{\bf Automating Resolution is NP-Hard}}

\author{
Albert Atserias \;\;\; and \;\;\;
Moritz M\"uller \\ Universitat Polit\`ecnica de Catalunya
}

\maketitle

\begin{abstract}
  We show that the problem of finding a Resolution refutation that is
  at most polynomially longer than a shortest one is NP-hard. In the
  parlance of proof complexity, Resolution is not automatizable unless
  P = NP.  Indeed, we show that it is NP-hard to distinguish between
  formulas that have Resolution refutations of polynomial length and
  those that do not have subexponential length refutations. This also
  implies that Resolution is not automatizable in subexponential time
  or quasi-polynomial time unless~NP is included in SUBEXP or QP,
  respectively. 
\end{abstract}

\section{Introduction}

The proof search problem for a given proof system asks, given a
tautology, to find an approximately shortest proof of it.  Clearly,
the computational complexity of such problems is of fundamental
importance for automated theorem proving. In particular, among the
proof systems for propositional logic, Resolution deserves special
attention since most modern implementations of satisfiability solvers
are based on it.

We say that the proof search problem for Resolution is solvable in
polynomial time if there is an algorithm that, given a contradictory
CNF formula~$F$ as input, outputs a Resolution refutation of~$F$ in
time polynomial in~$r+s$, where~$r$ is the size of~$F$, and~$s$ is the
length of a shortest Resolution refutation of~$F$.  More succinctly,
we say that Resolution is {\em automatizable}~\cite{bpr}. 
It is clear that the concept of automatizability applies not only to
Resolution but to any refutation or proof system, and one can ask for
automating algorithms that run in quasi-polynomial time,
subexponential time, etc..\footnote{The time of the automating
  algorithm is not measured in~$r$ but in~$r+s$ because~$s$ can be
  much larger than~$r$. We use both~$r$ and~$s$, and not just~$s$,
  because a Resolution refutation need not use all clauses in~$F$, but
  the algorithm should be given the opportunity to at least read all
  of~$F$.}

In this paper we show that Resolution is not automatizable
unless~$\PTIME=\NPTIME$. The assumption is clearly optimal
since~$\PTIME=\NPTIME$ implies that it is.  To prove our result we
give a direct and efficient reduction from~$\THREESAT$, the
satisfiability problem for 3-CNF formulas. The reduction is so
efficient that it also rules out quasi-polynomial and subexponential
time automating algorithms for Resolution under the corresponding
hardness assumptions. More precisely, let~$\QPTIME$ and~$\SUBEXP$
denote the classes of problems that are decidable in quasi-polynomial
time~$2^{(\log n)^{O(1)}}$, and in subexponential time~$2^{n^{o(1)}}$,
respectively. 
 Then our main result reads:

\begin{theorem}\label{thm:main}\
\begin{enumerate}\itemsep=0pt
\item Resolution is not automatizable in subexponential time unless $\NPTIME\subseteq\SUBEXP$. 
\item Resolution is not automatizable in quasi-polynomial time unless $\NPTIME \subseteq\QPTIME$.
\item Resolution is not automatizable in polynomial time unless $\NPTIME \subseteq\PTIME$.
\end{enumerate}
\end{theorem}

That Resolution is not automatizable in polynomial time has been known
under a stronger assumption from parameterized complexity theory,
using a more contrived reduction~\cite{ar}: we review the literature
below. The first two statements in Theorem~\ref{thm:main} give the
first evidence that Resolution is not automatizable in
quasi-polynomial or subexponential time. As in the third statement,
their assumptions are also optimal in that~$\NPTIME \subseteq\QPTIME$
and~$\NPTIME \subseteq \SUBEXP$ imply that~Resolution can be automated
in quasi-polynomial and subexponential time, respectively.

The main result as stated in Theorem~\ref{thm:main} is a direct
consequence of the fact, which we also prove, that the problem of
non-trivially approximating minimum proof length for Resolution is
NP-hard.  If for a CNF formula~$G$ we write~$r(G)$ for the size
of~$G$, and~$s(G)$ for the length of a shortest Resolution refutation
of~$G$, then we show: 

\begin{theorem}\label{thm:approx} 
There are reals~$c > 0$ and~$d > 0$ and a
polynomial-time computable function~$G$ that maps any~3-CNF
formula~$F$ to a CNF formula~$G(F)$ such that, for~$r = r(G(F))$
and~$s = s(G(F))$:
  \begin{enumerate}\itemsep=0pt 
  \item[(a)] if~$F$ is satisfiable, then~$s < r^{c}$;
  \item[(b)] if~$F$ is unsatisfiable, then~$s
    >2^{r^{1/{d}}}$. \end{enumerate}
\end{theorem}
\noindent Moreover,~$c$ and~$d$ can be chosen arbitrarily close to~$1$ and~$2$,
respectively, which means that it is NP-hard to approximate the
minimal Resolution refutation length to
within~$2^{r^{1/2-\epsilon}}$ for any $\epsilon > 0$.

\paragraph{Proof idea} 
%

An idea of how a map~$G$ as in Theorem~\ref{thm:approx} could be
defined is implicit
in~\cite{pudlak}. Pudl\'ak~\cite[Theorem~2]{pudlak} maps a formula~$F$
to~$\REF(F,s)$, for some~$s$ suitable for his context,
where~$\REF(F,s)$ is a~CNF formula 
whose clauses describe, in a natural way, the~Resolution
refutations of~$F$ of length~$s$. He used this function to show that
the canonical pair of Resolution is symmetric. In particular, he
showed that, if~$F$ is satisfiable, then~$\REF(F,s)$ has a short
Resolution refutation.  This refutation proceeds naturally by using a
satisfying assignment for~$F$ as a guide to find a true literal in
each line of the alleged refutation, line by line one after another,
until it gets stuck at the final empty clause.  Conversely, we would
like to show that, if~$F$ is unsatisfiable, then~$\REF(F,s)$ 
is hard for Resolution.  Intuitively, this should be the case:
refuting~$\REF(F,s)$ means proving a lower bound and ``our experience
rather suggests that proving lower bounds is difficult'' -- this is
what Pudl\'ak~\cite[Section~3]{pudlak} states about a similar formula
for strong proof systems.

However, even after considerable time and effort, we failed to prove a
Resolution length lower bound for~$\REF(F,s)$. We bypass the issue by
considering a \emph{harder} version~$\RREF(F,s)$ of~$\REF(F,s)$. The
harder~$\RREF(F,s)$ is obtained from \emph{relativizing}~$\REF(F,s)$
seen as the propositional encoding of a first-order formula with a
built-in linear order, following the general relativization technique
of Dantchev and Riis~\cite{dr}. When~$F$ is satisfiable, 
Pudl\'ak's upper bound for $\REF(F,s)$ goes
through to $\RREF(F,s)$, and the linear order is crucial in this.
On the other hand, a random restriction 
argument in the style
of~\cite{dr} reduces a length lower bound for~$\RREF(F,s)$ to a
certain width lower bound for~$\REF(F,s)$. The bulk of the current
work is in establishing this width lower bound for~$\REF(F,s)$, 
when~$F$ is unsatisfiable.  
It is proved by showing that, even if~$s =
s(n)$ has (not too slow) polynomial growth, the formulas~$\REF(F,s)$
and~$\REF(F,2^{n+1})$ are indistinguishable by inferences of bounded
width, where~$n$ is the number of variables in~$F$. Since every
unsatisfiable~CNF 
formula with~$n$ variables has a refutation of length~$2^{n+1}$, the
formula~$\REF(F,2^{n+1})$ is satisfiable, from which it follows
that~$\REF(F,s)$ does not have bounded-width refutations.

The technical device  that we use  
in the indistinguishability argument is a variant of the {\em conditions}
from~\cite{ms}, a particular formalization of a Prover-Adversary
argument as,~e.g.,~in~\cite{dr}. The wording is meant to point out
some analogy with forcing conditions~\cite{amu}. This is not
straightforward. The main obstacle overcome by our variant is the
presence of the built-in linear order in~$\REF(F,s)$.  In fact,
Dantchev and Riis~\cite[Section~5]{dr} point out explicitly that their
arguments fail in the presence of a built-in linear order.

\paragraph{History of the problem} The complexity of the proof search
problem has been extensively investigated.  Kraj\'i\v{c}ek and
Pudl\'ak \cite{kp} showed that Extended Frege systems\footnote{We
  refer to the textbook~\cite[Chapter~4]{krabuch} for a definition of
  this and the following systems. All notions relevant to state and
  prove our results are going to be defined later.} are not
automatizable assuming RSA is secure against~$\Ppoly$. Subsequently,
Bonet et al. showed this for Frege~\cite{bpr} and bounded depth Frege
systems \cite{bdgmp} assuming the Diffie-Hellman key exchange is
secure against polynomial or, respectively, subexponential size
circuits.

In fact, these results rule out feasible interpolation, an influential
concept introduced to proof complexity by
Kraj\'i\v{c}ek~\cite{krabd,kraint}. We refer to \cite[Chapters
  17,~18]{kraneu} for an account.  If~a system with feasible
interpolation has short refutations of the contradictions that state
that a pair of NP problems are not disjoint, then the pair can be
separated by small circuits. Hence, feasible interpolation can be
ruled out by finding short proofs of the disjointness of an NP pair
that is hard to separate. Such hardness assumptions turn up naturally
in cryptography~\cite{groll} which explains the type of assumptions
that were used in the results above.

The failure of feasible interpolation for a natural system~$R$ implies
(cf.~\cite[Theorem~3]{ab}) that~$R$ is not even {\em weakly}
automatizable in the sense that it would be polynomially simulated
(see~\cite{cookreck}) by an automatizable system.  Hence, the above
results left open whether weak proof systems, in particular those
having feasible interpolation such as Resolution~\cite{kraint}, were
(weakly) automatizable.
We refer to~\cite{ats} for a survey, and focus from now on on
Resolution.

Pudl\'ak showed~\cite[Corollary~2]{pudlak} that the weak
automatizability of a proof system is equivalent to the (polynomial
time) separability of its, so-called, canonical
NP-pair~\cite{razNP}. This is, informally, the feasibility of
distinguishing between satisfiable formulas and those with short
refutations. Hence, to rule it out it suffices to reduce some
inseparable disjoint NP pair to it. Atserias and Maneva~\cite{am}
found in this respect useful pairs associated to two player games. The
two NP sets collect the games won by the respective players, and
separation means deciding the game.  Following~\cite{am,HP}, Beckmann
et al.~\cite{bpt} showed that Resolution is not weakly automatizable
unless parity games are decidable in polynomial time. Note, however,
that this might well be the case, in fact, parity games are decidable
in quasi-polynomial time~\cite{parityquasi}.

Moreover, some non-trivial automating algorithms are known.  Beame and
Pitassi~\cite{bp} observed that treelike Resolution \emph{is}
automatizable in quasi-polynomial time.
For general Resolution there is an algorithm that, when given a 3-CNF
formula with~$n$ variables that has a Resolution refutation of length
at most~$s$, computes a refutation in time~$n^{O(\sqrt{n\log s})}$.
%
%
This follows from the size-width trade-off of Ben-Sasson and
Wigderson~\cite{bw}.  Indeed, it is trivial to find a refutation of
width at most~$w$ in time~$n^{O(w)}$ if there is one (and, in general,
time~$n^{\Omega(w)}$ is necessary~\cite{atsnord}).  When~$s$
is subexponential the runtime of this algorithm is the
non-trivial~$2^{n^{1/2+o(1)}}$. 

However, the automatizability of Resolution is unlikely. First,
Alekhnovich et al.~\cite{abmp} showed, assuming
only~$\PTIME\neq\NPTIME$, that automatization is not possible in
linear time. In fact, they proved more. They considered the
optimization problem of finding, given a contradictory~CNF, a
Resolution refutation that is as short as possible.  They reduced to
it the optimization problem MMCSA of finding, given a monotone
circuit, a satisfying assignment that has Hamming weight as small as
possible. Known PCP theorems imply that this problem is not
approximable with superconstant but sublinear ratio~$2^{\log^{1-o(1)} n}$, so the same holds for
finding short Resolution refutations. This argument can be adapted to
many other refutation systems (see~\cite{abmp}).

But the main convincing evidence that Resolution is not automatizable,
before the result of this paper, was achieved by Alekhnovich and
Razborov~\cite{ar}. By a different and ingenious reduction they showed
that if Resolution, or even treelike Resolution, were automatizable,
then
MMCSA would have, in the terminology of parameterized complexity
theory (see~\cite[Proposition~5]{cgg}), an fpt algorithm with constant
approximation ratio. Now, the same paper~\cite{ar} also established
``the first nontrivial parameterized inapproximability
result''~\cite[p.9]{eick} by further deriving a randomized fpt
algorithm for the parameterized decision version of MMCSA, a
well-known~$\WP$-complete problem (see
e.g.~\cite[Theorem~3.14]{flumgrohe}).  The randomized fpt algorithm
has subsequently been derandomized by Eickmeyer et al.~\cite{eick},
hence Resolution is not automatizable unless~$\WP=\FPT$.  Very
recently, Mertz et al.~\cite{mertz} showed that Resolution is not
automatizable in time~$n^{(\log\log n)^{0.14}}$ unless ETH fails; this
follows the same line of argument as \cite{ar} but is based on a more
recent parameterized inapproximability result due to Chen and
Lin~\cite{chenlin}.

Since these results apply not only to Resolution but even to treelike
Resolution, which is automatizable in quasipolynomial time,
Alekhnovich and Razborov stated that the ``main problem left open''
\cite[Section~5]{ar} is whether general Resolution is automatizable in
quasi-polynomial time.  We consider Theorem~\ref{thm:main} as an
answer to this question.


The computational problem of computing minimal proof lengths also has
a long history.  For first-order logic, the problem dates back to
G\"odel's famous letter to von Neumann; we refer to~\cite{urq} for a
historical discussion, to~\cite{buss} for a proof of G\"odel's claim
in the letter, and to~\cite{cheflugoe} for some more recent results.
In propositional logic, the problem has been shown to be NP-hard for a
particular Frege system by Buss~\cite{buss}, and for Resolution by
Iwama~\cite{iwama}.  Alekhnovich et al.~\cite{abmp} showed that the
minimal Resolution refutation length cannot be approximated to within
any fixed polynomial unless~$\NPTIME\not\subseteq\Ppoly$: for
every~$d\in\N$ there are functions~$G$ and~$S$, computable in
non-uniform polynomial time, such that for every CNF formula~$F$ of
sufficiently large size~$r = r(F)$ we have either~$s(G(F))<S(r)$ 
or~$s(G(F))>S(r)^d$
depending on the satisfiability~$F$.  This falls short to rule out
automatizability because~$S(r)$ has exponential growth.  Earlier,
Iwama~\cite{iwama} found uniformly computable such functions with
polynomially bounded~$S(r)$ but his gap was only~$S(r)$
versus~$S(r)+r^d$ for a constant~$d$,
so also falls short to rule out
automatizability.

\paragraph{Outline} In
Section~\ref{sec:preliminaries} we introduce some notation and basic
terminology from propositional logic. Section~\ref{sec:refutations}
presents Resolution refutations as finite structures.
Section~\ref{sec:ref} is devoted to~$\REF(F,s)$ and proves the width
lower bound when~$F$ is unsatisfiable
(Lemma~\ref{lem:indexwidthlowerbound}). Section~\ref{sec:rref}
discusses the relativized formula~$\RREF(F,s)$, the refutation length
upper bound when~$F$ is satisfiable (Lemma~\ref{lem:upperbound}), and
the refutation length lower bound when~$F$ is unsatisfiable
(Lemma~\ref{lem:lowerbound}).  Theorems~\ref{thm:approx}
and~\ref{thm:main} are derived from these lemmas in
Section~\ref{sec:hardness}.  In Section~\ref{sec:conclusions} we
discuss some open issues. Finally, for easiness of reference, in
Appendix~\ref{sec:appendix} we give the detailed lists of clauses for
the formulas~$\REF$ and~$\RREF$.

\section{Preliminaries} \label{sec:preliminaries}

For~$n\in\N$ we let~$[n]:=\{1,\ldots,n\}$ and understand
that~$[0]=\emptyset$.  A {\em partial function} from a set~$A$ to a
set~$B$ is a function~$f$ with domain~$\Dom(f)$ included in~$A$ and
image~$\Img(f)$ included in~$B$.  We view partial functions from~$A$
to~$B$ as sets of ordered pairs~$(u,v) \in A \times B$. For any
set~$C$, the {\em restriction of~$f$ to~$C$}
is~$f\cap(C\times\Img(f))$. The {\em restriction of~$f$ with
  image~$C$} is~$f\cap(\Dom(f)\times C)$.
 
We fix some notation for propositional logic. Let~$\Vars$ be a set of
propositional variables that take truth values in~$\Bool = \{0,1\}$,
where~$0$ denotes {\em false} and~$1$ denotes~{\em true}. A {\em
  literal} is a variable~$X$ or its negation~$\neg X$, also
denoted~$\bar X$. We also write~$X^{(1)}$ for~$X$ and~$X^{(0)}$
for~$\bar X$. A {\em clause} is a set of literals, that we write as a
disjunction of its elements.  A clause is {\em non-tautological} if it
does not contain both a variable and its negation.  The {\em size} of
a clause is the number of literals in it.  A~{\em CNF~formula},
or~{\em CNF}, is a set of clauses, that we write as a conjunction of
its elements. A~$k$-CNF, where~$k \geq 1$, is a CNF in which all
clauses have size at most~$k$. The {\em size} of a CNF~$F$ is the sum
of the sizes of its clauses. We use~$r(F)$ to denote the size
of~$F$.

An {\em assignment}, or {\em restriction}, is a partial map from the
set of variables~$\Vars$ to~$\Bool$. If~$\alpha$ is an assignment
and~$X^{(b)}$ is a literal, then~$\alpha$ {\em satisfies}~$X^{(b)}$
if~$X \in \Dom(\alpha)$ and~$b = \alpha(X)$; it {\em
  falsifies}~$X^{(b)}$ if~$X \in \Dom(\alpha)$ and~$b =
1-\alpha(X)$. If~$C$ is a clause, then~$\alpha$ {\em satisfies}~$C$ if
it satisfies some literal of~$C$;
it {\em falsifies}~$C$ if it falsifies every literal of~$C$.
The {\em restriction of~$C$ by~$\alpha$}, denoted~$C\res\alpha$,
is~$1$ if~$\alpha$ satisfies~$C$ and~$0$ if~$\alpha$ falsifies~$C$;
if~$\alpha$ neither satisfies nor falsifies~$C$, then~$C\res\alpha$ is
the clause obtained from~$C$ by removing all the falsified literals
of~$C$, i.e.,~$C\res\alpha = C\setminus\{ X^{(1-\alpha(X))}\mid
X\in\Dom(\alpha)\}$.  If~$F$ is a CNF, then~$F\res\alpha$ is the CNF
that contains~$C\res\alpha$ for those~$C\in F$ which are neither
satisfied nor falsified by~$\alpha$, and that contains the empty
clause if some~$C\in F$ is falsified by~$\alpha$.

A clause~$D$ is a {\em weakening} of clause~$C$ if~$C \subseteq D$.  A
clause~$E$ is a {\em resolvent} of clauses~$C$ and~$D$ if there is a
variable~$X$ such that~$X\in C$ and~$\bar X\in D$, and~$E =
(C\setminus\{X\})\cup(D\setminus\{\bar X\})$; we then speak of the
resolvent of~$C$ and~$D$ {\em on~$X$}, that we denote
by~$\mathrm{res}(C,D,X)$. We also say that~$E$ is obtained from~$C$
and~$D$ by a {\em cut on~$X$}.

Let~$F$ be a CNF. A {\em Resolution proof from~$F$} is a
sequence~$(D_1,\ldots,D_s)$ of non-tautological clauses, where~$s \geq
1$ and, for all~$u\in[s]$, it holds that~$D_u$ is a weakening of a
clause in~$F$, or there are~$v,w\in[u-1]$ such that~$D_u$ is a
weakening of a resolvent of~$D_v$ and~$D_{w}$. The {\em length} of the
proof is~$s$; each~$D_u$ is a {\em line}. A {\em Resolution refutation
  of~$F$} is a proof from~$F$ that ends with the empty clause,
i.e.,~$D_s = \emptyset$.  We let~$s(F)$ denote the minimal~$s$ such
that~$F$ has a Resolution refutation of length~$s$; if~$F$ is
satisfiable, we let~$s(F)=\infty$. For a sequence of
clauses~$\Pi=(D_1,\ldots,D_s)$ let~$\Pi\res\alpha$ be obtained
from~$(D_1\res\alpha,\ldots, D_s\res\alpha)$ by removing 1's and
replacing 0's by the empty clause. It is clear that if~$\Pi$ is a
Resolution refutation of~$F$ of length~$s$, then~$\Pi\res\alpha$ is a
Resolution refutation of~$F\res\alpha$ of length at most~$s$.

%

\section{Refutations as structures} \label{sec:refutations}

For this section we fix a CNF~$F$ with~$n$ variables~$X_1,\ldots,X_n$
and~$m$ clauses~$C_1,\ldots,C_m$. We view Resolution
refutations~$(D_1,\ldots, D_s)$ of~$F$ of length~$s$ as finite
structures with a ternary relation~$D$ and four unary
functions~$V,I,L,R$:
\begin{equation}\label{eq:reftype}
\begin{split}
& D \subseteq [s] \times [n] \times \Bool, \\
& V : [s] \rightarrow [n] \cup \{0\}, \\ 
& I : [s] \rightarrow [m] \cup \{0\}, \\
& L : [s] \rightarrow [s] \cup \{0\}, \\
& R : [s] \rightarrow [s] \cup \{0\}.
\end{split}
\end{equation}
The meaning of~$(u,i,b) \in D$ is that the literal~$X_i^{(b)}$ is
in~$D_u$.  For each~$u\in[s]$ exactly one of~$V(u)$ or~$I(u)$ is
non-zero. The meaning of~$V(u)=i\in[n]$ is that~$D_u$ is a weakening
of the resolvent of~$D_v$ and~$D_w$ on~$X_i$, where~$v = L(u)\in[u-1]$
and~$w = R(u)\in[u-1]$, and~$\bar X_i\in D_v$ and~$ X_i\in D_w$. The
meaning of~$I(u)=j\in[m]$ is that~$D_u$ is a weakening of the
clause~$C_j$ of~$F$.  Formally, a structure~$(D,V,I,L,R)$ of type
\eqref{eq:reftype} is a {\em refutation of~$F$ of length~$s$} if the
following hold for all~$u,v \in [s]$,~$i,i' \in [n]$,~$j \in [m]$,
and~$b \in \Bool$:
%
%

\medskip

\begin{longtable}{cll}
(R1) & & $V(u) = 0$ or $I(u) = 0$, but not both; \\
(R2) & & if $I(u)=0$, then both $R(u)\neq 0$ and $L(u)\neq 0$; \\
(R3) & & $L(u) < u$ and $R(u) < u$; \\
(R4a) & & if $V(u)=i$ and $L(u)=v$, then $(v,i,0)\in D$; \\
(R4b) & & if $V(u)=i$ and $R(u)=v$, then $(v,i,1)\in D$; \\
(R5a) & & if $V(u)=i\not=i'$, $L(u)=v$, and $(v,i',b)\in D$, then $(u,i',b) \in D$; \\
(R5b) & & if $V(u)=i\not=i'$, $R(u)=v$, and $(v,i',b)\in D$, then $(u,i',b) \in D$; \\
(R6) & & if $I(u)=j$ and $X^{(b)}_i$ appears in $C_j$, then $(u,i,b)\in D$; \\
(R7) & & $(u,i,0) \not\in D$ or $(u,i,1) \not\in D$; \\
(R8) & & $(s,i,b)\notin D$.
\end{longtable}
\medskip

In words, (R1) determines, for every line~$D_u$, whether it is a
weakening of an initial clause, i.e.,~$I(u)\neq 0$, or a weakening of
a resolvent, i.e.,~$V(u)\neq 0$. In the first case~$C_{I(u)}\subseteq
D_u$ by~(R6). In the second
case,~$\mathrm{res}(D_{L(u)},D_{R(u)},X_{V(u)}) \subseteq D_u$ by (R4)
and (R5), with (R2) and~(R3) ensuring that~$D_{L(u)}$ and~$D_{R(u)}$
are earlier lines in the sequence.
Finally, (R7) ensures no~$D_u$ is tautological, and (R8) ensures~$D_s$
is empty.

We give an example that will play a crucial role in the proof of the
width lower bound.

\begin{example} \label{ex:full}
We use $(D^*,V^*,I^*,L^*,R^*)$ to denote the \emph{full-tree Resolution
  refutation} of $F$. It has length 
\begin{equation*}
  s^*:=2^{n+1}-1
\end{equation*}
  and its
clauses are arranged in the form of a full binary tree of height $n$
with $2^n-1$ internal nodes and $2^n$ leaves. This tree has
one node $n_a$ at level $h \in \{0\} \cup [n]$ for every $a =
(a_1,\ldots,a_h) \in \{0,1\}^h$ that is labelled by the clause 
\begin{equation*}
C_a= X_1^{(a_1)}\vee\cdots\vee X_h^{(a_h)},
\end{equation*}
that is, the unique clause in these variables falsified by the
assignment that maps~$X_i$ to~$1-a_i$.
In particular, the root of the tree is labelled by the empty clause
and, for~$h\in[n]$ and~$a \in \{0,1\}^{h-1}$, the clause~$C_a$ that
labels node~$n_a$ is the resolvent of the clauses~$C_{a1}$
and~$C_{a0}$ that label the children nodes~$n_{a1}$ and~$n_{a0}$ on
the variable~$X_h$, i.e.,~$C_a =
\mathrm{res}(C_{a1},C_{a0},X_h)$. Since~$F$ is unsatisfiable, every
clause~$C_a$ that labels a leaf~$n_a$ is a weakening of some
clause~$C_j$ of~$F$.

To view this refutation as a structure of type \eqref{eq:reftype} we
have to identify the nodes~$n_a$ with numbers in~$[s^*]$. We first
identify the leafs, i.e., the nodes~$n_a$ with~$a\in\{0,1\}^n$, with
the numbers~$[2^n]$, then we identify the nodes on level~$n-1$, i.e.,
the nodes~$n_a$ with~$a\in\{0,1\}^{n-1}$, with the numbers
in~$[2^{n}+2^{n-1}]\setminus[2^{n}]$ and so on, with the root
getting~$s^*=2^{n+1}-1$.

Let~$a=(a_1,\ldots,a_h)\in\{0,1\}^h$ for~$h\le n$. We
set~$V^*(n_a):=0$ if~$h=n$, and~$V^*(n_a):=h$ if~$h<n$. We
set~$I^*(n_a):=0$ if~$h<n$, and~$I^*(n_a):=j$ if~$h=n$ and~$j\in[m]$
is, say, smallest such that~$C_a$ is a weakening of~$C_j$.  We
set~$L^*(n_a):=R^*(n_a):=0$ if~$h=n$. If~$h<n$ we
set~$L^*(n_a):=n_{a0}$ and~$R^*(n_a):=n_{a1}$.  Finally,~$(n_a,i,b)\in
D^*$ if and only if~$i\in[h]$ and~$b=a_i$.

\end{example}

\section{Non-relativized formula REF} \label{sec:ref}

Given a CNF~$F$ with~$n$ variables~$X_1,\ldots,X_n$ and~$m$
non-tautological clauses~$C_1,\ldots,C_m$, and a natural number~$s
\geq 1$, we describe a CNF formula~$\REF(F,s)$ that is satisfiable if
and only if~$F$ has a refutation of length~$s$. Its variables are:
\begin{itemize} \itemsep=0pt
\item $D[u,i,b]$ for~$u \in [s]$,~$i \in [n]$,~$b \in \Bool$
  indicating that~$(u,i,b) \in D$.
\item $V[u,i]$ for~$u \in [s]$,~$i \in [n] \cup \{0\}$ indicating
  that~$V(u) = i$.
\item $I[u,j]$ for~$u \in [s]$,~$j \in [m] \cup \{0\}$ indicating
  that~$I(u) = j$.
\item $L[u,v]$ for~$u \in [s]$,~$v \in [s]\cup\{0\}$ indicating
  that~$L(u) = v$.
\item $R[u,v]$ for~$u \in [s]$,~$v \in [s]\cup\{0\}$ indicating
  that~$R(u) = v$.
\end{itemize}

Clearly, any assignment to these variables describes a ternary
relation~$D$ and binary relations~$V$,~$I$,~$L$ and~$R$.  The clauses
of~$\REF(F,s)$ are listed in Table~\ref{fig:refclauses} of
Appendix~\ref{sec:appendix}.  This set of clauses is satisfied
precisely by those assignments that describe refutations of~$F$ of
length~$s$. Conversely, given a structure as in \eqref{eq:reftype} the
{\em associated} assignment~$\alpha$ satisfies~$\REF(F,s)$ if and only
if~$(D,V,I,L,R)$ is a refutation of~$F$ of length~$s$; this
assignment~$\alpha$ maps
variables~$D[u,i,b]$,~$V[u,i]$,~$I[u,j]$,~$L[u,v]$, and~$R[u,v]$ to 1
or 0 depending on whether, respectively,~$(u,i,b)\in
D$,~$V(u)=i$,~$I(u)=j$,~$L(u)=v$, and~$R(u)=v$ or not.

The index $u \in [s]$ is \emph{mentioned} in the variables 
\begin{equation*}
D[u,i,b], V[u,i], I[u,j], L[u,v], R[u,v].
\end{equation*}
Observe that if~$v \not= u$, then~$v$ is not mentioned in~$L[u,v]$
or~$R[u,v]$. The \emph{index-width} of a clause is the number of
indices mentioned by some variable occurring in the clause.  Observe
that all clauses of~$\REF(F,s)$ have index-width at most two.  The
index-width of a Resolution refutation is the maximum index-width of
its clauses.

\begin{lemma} \label{lem:indexwidthlowerbound}
For all integers~$n,w,s \geq 1$ with~$2^n \ge s \geq 6nw$ and every
unsatisfiable CNF~$F$ with~$n$ variables, every Resolution refutation
of~$\REF(F,s)$ has index-width at least~$w$.
\end{lemma}

\begin{proof}
Fix an unsatisfiable CNF~$F$ with~$n$ variables and~$m$ clauses.  For
this proof let~$G$ denote the formula~$\REF(F,s)$ and let~$G^*$ denote
the formula~$\REF(F,s^*)$, where~$s^* = 2^{n+1} - 1$ is the length of
the full-tree Resolution refutation of~$F$ from Example~\ref{ex:full},
which exists for~$F$ because it is unsatisfiable.  Let~$\alpha^*$ be
the assignment associated to~$(D^*,V^*,I^*,L^*,R^*)$.

Let~$k$ be an integer such that~$2^k < 3w \le 2^{k+1}$ and note
that~$1 \leq k < n$ since~$n,w \geq 1$ and~$2^n \geq 6nw$.  We
partition~$[s^*]$ into~$n-k+1$
intervals~$B^*_0,B^*_1,\ldots,B^*_{n-k}$ where
\begin{align*}
& B^*_0 := [s^*]\setminus [s^*-2^{k+1}+1],\\
& B^*_i := [s^*-2^{k+i}+1]\setminus [s^*-2^{k+1+i}+1]\quad\textup{ for $i = 1,\ldots,n-k$.}
\end{align*}
In the notation of Example~\ref{ex:full},~$B^*_0=\{n_a\mid
a\in\{0,1\}^{\le k}\}$ is the set of~$2^{k+1}-1$ many nodes at the
top~$k$ levels of the full binary tree. For~$i \in [n-k]$, the~$i$-th
block~$B^*_i=\{n_a\mid a\in\{0,1\}^{k+i}\}$ is the set of nodes at
level~$k+i$ of the full binary tree. In particular,~$B^*_{n-k}$ is the
set of leaves.

Likewise, we partition~$[s]$ into~$n-k+1$
intervals~$B_0,B_1,\ldots,B_{n-k}$ where
\begin{align*}
& B_0 := [s]\setminus[s-2^{k+1}+1],\\ 
& B_i := [s-2^{k+1} \cdot i+1]\setminus [s-2^{k+1} \cdot 
(i+1)+1]\quad\textup{ for $i
= 1,\ldots,n-k-1$,}\\
& B_{n-k} := [s-2^{k+1}\cdot (n-k)+1]. 
\end{align*} 
Observe that~$|B^*_0|=|B_0|=2^{k+1}-1$; let~$t : B_0 \rightarrow
B^*_0$ be the bijection defined by~$t(u) := u-s+s^*$ so that for
all~$u,v\in B_0$ it holds that
\begin{equation}\label{eq:t}
u< v \;\;\text{ if, and only if, }\;\; t(u)< t(v).
\end{equation}
Observe that for all~$i\in[n-k-1]$:
\begin{align}
& |B^*_i| = 2^{k+i} \geq 2^{k+1}= |B_i| \ge 3w. \label{eq:B1} \\
& |B^*_{n-k}| = 2^n \geq s-2^{k+1}\cdot (n-k)+1 = |B_{n-k}| \geq 3w,
\label{eq:B2}
\end{align}
with the second following from~$2^n\ge s \geq 6nw$ and~$1\leq k < n$.

Let~$\mathscr{H}$ be the collection of partial
functions~$h:[s]\cup\{0\}\rightarrow [s^*]\cup\{0\}$ such that:

\medskip
\begin{longtable}{cll}
(H1) & & $h$ is injective, \\
(H2) & & $0 \in \Dom(h)$ and $h(0)=0$, \\
(H3) & &  if $u \in \Dom(h) \cap B_0$, then $h(u) = t(u) \in B^*_0$, \\
(H4) & & if $u \in \Dom(h) \cap B_i$ with $i \in [n-k]$, then $h(u) \in B^*_i$.
\end{longtable}
\medskip

\noindent In words, condition (H4) says that~$h$ preserves membership
in matching intervals, and (H3) says that the 0-intervals are kept
intact through the fixed bijection~$t$. Preserving the intervals has
the following important consequence:

\begin{claim} \label{claim:obvious}
For every~$h \in \mathscr{H}$ and~$u,v \in \Dom(h)\setminus\{0\}$ the
following hold:
\begin{enumerate} \itemsep=0pt
\item $h(u) \not= 0$ and $h(v) \not= 0$,
\item if $L^*(h(v)) \in \Img(h)$, then $h^{-1}(L^*(h(v))) < v$,
\item if $R^*(h(v)) \in \Img(h)$, then $h^{-1}(R^*(h(v))) < v$.
\end{enumerate}
\end{claim}

\begin{proof} Property \emph{1} follows from (H1) and (H2).  To prove
  \emph{2} we distinguish several cases: If~$v \in B_{n-k}$,
  then~$h(v) \in B^*_{n-k}$ by~(H4), hence~$L^*(h(v)) = 0$
  and~$h^{-1}(L^*(h(v))) = 0$ by~(H2), which is smaller than~$v \not=
  0$.  If~$v \in B_i$ for some~$i \in [n-k-1]$, then~$h(v) \in B^*_i$
  by~(H4), hence~$L^*(h(v)) \in B^*_{i+1}$, and~$h^{-1}(L^*(h(v))) \in
  B_{i+1}$ by~(H4) again, which is smaller than~$v \in B_i$.  If~$v\in
  B_0$, then first note that~$h(v)=t(v)\in B^*_0$ by (H3).  We
  distinguish the cases whether~$L^*(h(v)) \in B_0^*$ or not.  In
  case~$L^*(h(v))\in B_0^*$, we
  have~$h^{-1}(L^*(h(v)))=t^{-1}(L^*(h(v)))$. Since~$L^*(h(v))<h(v)$,
  by~\eqref{eq:t} we have~$t^{-1}(L^*(h(v)))<t^{-1}(h(v))=t^{-1}(t(v))
  = v$.  In case~$L^*(h(v))\notin B_0^*$, we have~$L^*(h(v))\in
  B^*_1$, so~$h^{-1}(L^*(h(v)))\in B_1$ by~(H4), which is smaller
  than~$v\in B_0$.  The proof of \emph{3} is analogous to that of
  \emph{2}.
\end{proof}

For a set~$I \subseteq [s^*]\cup\{0\}$, let
\begin{equation*}
\partial I := \big\{
L^*(u) \mid u \in I\setminus \{0\} \big\} \cup \big\{ R^*(u) \mid u \in I\setminus\{0\} \big\}.
\end{equation*}
A \emph{condition} is a pair~$p = (g,h)$, where~$g$ and~$h$ are
functions in~$\mathscr{H}$, such that

\medskip
\begin{longtable}{cll}
(C1) & & $g \subseteq h$, \\
(C2) & & $\Img(h) = \Img(g) \cup \partial \Img(g)$.
\end{longtable}
\medskip

\noindent We say a condition~$p' = (g',h')$ {\em extends}~$p$
if~$h\subseteq h'$, i.e.,~$h'$ extends~$h$ as a function.  Observe,
since~$0\in\Dom(g)$,
\begin{equation}\label{eq:sizeh}
|\Dom(h)|\le 3|\Dom(g)|-2.
\end{equation}

%
%
We define a partial truth assignment~$\alpha(p)$ that sets the
variables of~$G$ as follows. Note that if~$D[u,i,b]$,~$V[u,i]$,
and~$I[u,j]$ are variables of~$G$, then~$D[g(u),i,b]$,~$V[g(u),i]$,
and~$I[g(u),j]$ are variables of~$G^*$ which are evaluated
by~$\alpha^*$. The assignment~$\alpha(p)$ is defined precisely on the
variables of~$G$ that mention some~$u\in \Dom(g)$.  For such~$u$ it
maps
\begin{enumerate} \itemsep=0pt
\item[--] $D[u,i,b]$ to $ \alpha^*(D[g(u),i,b])$, for all $i \in [n]$ and $b\in\Bool$;
\item[--] $V[u,i]$ to $\alpha^*( V[g(u),i])$, for all $i \in [n]\cup\{0\}$;
\item[--] $I[u,j]$ to $\alpha^*(  I[g(u),j])$, for all $j \in [m]\cup\{0\}$;
\item[--] $L[u,v]$ to $1$ or $0$ indicating whether $v = h^{-1}(L^*(g(u)))$, for all $v \in [s] \cup\{0\}$;
\item[--] $R[u,v] $ to $ 1$ or $0$ indicating whether $v = h^{-1}(R^*(g(v)))$, for all $v \in [s] \cup \{0\}$.
\end{enumerate}
Note that~$L^*(g(u))$ and~$R^*(g(u))$ belong to~$\partial\Img(g)
\subseteq \Img(h)$ for every~$u \in \Dom(g)$, so~$h^{-1}$ is defined
in the last two cases.

Clearly, if a condition~$p'$ extends~$p$,
then~$\alpha(p)\subseteq\alpha(p')$.  For~$I \subseteq [s]$, the {\em
  restriction of~$p$ to~$I$}, denoted~$p\res I$, is the
pair~$(g^*,h^*)$ where~$g^*$ is the restriction of~$g$ to~$I \cup
\{0\}$, and~$h^*$ is the restriction of~$h$ with image~$\Img(g^*) \cup
\partial\Img(g^*)$.

\begin{claim}\label{cl:res}
If~$p$ is a condition and~$I \subseteq [s]$, then~$p\res I$ is a
condition and~$\alpha(p\res I)\subseteq\alpha(p)$.
\end{claim}

\begin{proof}
The requirement that~$g'$ and~$h'$ belong to~$\mathscr{H}$ is
obviously satisfied since (H1)-(H4) are preserved by restrictions to
subsets that contain~$0$. (C1) and (C2) are clear, so~$p\res I$ is a
condition. The inclusion~$\alpha(p\res I)\subseteq\alpha(p)$ is clear
since~$p$ extends~$p\res I$.
\end{proof}

\begin{claim}\label{cl:extend}
If~$p=(g,h)$ is a condition with~$|\Dom(g)|\le w$ and~$u\in[s]$, then
there exists a condition~$p'=(g',h')$ that extends~$p$ and such
that~$\Dom(g')=\Dom(g)\cup\{u\}$.
\end{claim}

\begin{proof}
We assume~$u\not\in\Dom(g)$ (otherwise we take~$p':=p$) and set~$g':=
g\cup\{(u,u')\}$ for~$u'\in[s^*]$ chosen as follows: if~$u\in B_0$,
take~$u':=t(u)$; otherwise~$u\in B_i$ for some~$i \in [n-k]$ and we
choose~$u'\in B^*_i\setminus\Img(h)$. Note there exists~$u'$ as
desired because~$ |B^*_i|\ge 3w$ by~\eqref{eq:B1} or~\eqref{eq:B2}, so
by~\eqref{eq:sizeh}
\begin{equation*}
|B^*_i\setminus\Img(h)| \ge |B^*_i|-3\cdot |\Dom(g)| +2>0.
\end{equation*}
It is clear that~$g'\in\mathscr{H}$. Write~$v'_0:=L^*(u')$
and~$v'_1:=R^*(u')$. We have to find~$v_0,v_1\in[s]\cup\{0\}$ such
that~$h':=h\cup\{(v_0,v_0'),(v_1,v_1')\}\in\mathscr H$. Assume at
least one of~$v_0',v'_1$ is not in~$\Img(h)$. Then it is distinct
from~$0$ (i.e.,~$u'\notin B^*_{n-k}$), say it is in~$B_i^*$. If~$i=0$,
we find~$v_0,v_1$ as the pre-images of~$v_0',v_1'$
under~$t$. Otherwise~$i\in[n-k]$ and we choose~$v_0,v_1\in B_i$ such
that~$h'$ is injective. This can be done because~$ |B_i|\ge 3w$
by~\eqref{eq:B1} or~\eqref{eq:B2}, so by~\eqref{eq:sizeh}
\begin{equation*}
|B_i\setminus\Dom(h)| \ge |B_i|-3\cdot |\Dom(g)| +2\ge 2.
\end{equation*}
It is clear that~$h'\in\mathscr H$.
\end{proof}

\begin{claim}\label{cl:axtrue}
If~$p$ is a condition and~$C$ is a clause of~$G$, then~$C\res
\alpha(p) \not= 0$.
\end{claim}

\begin{proof}
Let~$p = (g,h)$, write~$\alpha:=\alpha(p)$ and assume~$\alpha$ is
defined on all variables of~$C$. Then~$g$ is defined on all indices
mentioned by~$C$. We distinguish by cases according to the type
(A1)-(A21) of~$C$.
\begin{enumerate}\itemsep=0pt
\item[--] In case $C$ is of type (A1), i.e.,
  $C=\bigvee_{i\in[n]\cup\{0\}}V[u,i]$ for some $u\in\Dom(g)$, then
  $C\res \alpha$ equals $(\bigvee_{i\in[n]\cup\{0\}}V[g(u),i])\res \alpha^*$
  and this is 1 because $\bigvee_{i\in[n]\cup\{0\}}V[g(u),i]$ is a clause of~$G^*$.  Case (A2) is similar.
\item[--] In case (A3), ($u\in\Dom(g)$ and) $\alpha$ satisfies $L[u,v]$ for $v:=h^{-1}(L^*(g(u)))\in[s]\cup\{0\}$. 
Note $L^*(g(u))\in\partial\Img(g)\subseteq\Img(h)$, so $v$ is well-defined. Hence $C\res \alpha=1$.
Case (A4) is similar.
\item[--] In case (A5), $C\res \alpha$ equals $(\bar V[g(u),i]\vee \bar V[g(u),i'])\res \alpha^*$ and this is 1 because
$\bar V[g(u),i]\vee \bar V[g(u),i']$ is a clause of $G^*$. Case (A6) is similar.
\item[--] In case (A7), $v$ or $v'$ is distinct from $h^{-1}(L^*(g(u)))$ and then, respectively, $L[u,v]$ or $L[u,v']$ is falsified by $\alpha$.
Hence $C\res \alpha=1$. Case (A8) is similar.
\item[--] In case (A9), $C\res \alpha$ equals $(\bar I[g(u),0]\vee\bar V[g(u),0])\res \alpha^*$. But this is 1 since $\bar I[g(u),0]\vee\bar V[g(u),0]$ is a clause of $G^*$. Case (A10) is similar.
\item[--] In case  (A11), note $\alpha(L[u,0])=1$ implies $h^{-1}(L^*(g(u)))=0$, so $L^*(g(u))=0$ by Claim~\ref{claim:obvious}~(1). Then 
$g(u)$ is a leaf and $I^*(g(u))\neq 0$. Hence $0=\alpha^*(I[g(u), 0])=\alpha(I[u,0])$, so $C\res \alpha=1$.
Case (A12) is similar.
\item[--] In case (A13), note $\alpha(L[u,v])=1$ implies $v=h^{-1}(L^*(g(u)))$. But $h^{-1}(L^*(g(u)))=h^{-1}(L^*(h(u)))<u$ by (C1) and Claim~\ref{claim:obvious}~(2). Case (A14) is similar.
\item[--] In case (A15), $C\res \alpha=0$ implies  $u,v\in\Dom(g)$ and $v=h^{-1}(L^*(g(u)))$. Hence 
$h(v)=g(v)=L^*(g(u))$ (by (C1)) and $\alpha^*(L[g(u),g(v)])=1$. Further, $C\res  \alpha=0$ implies $\alpha^*(V[g(u),i])=1$ and $\alpha^*(D[g(v),i,0])=0$. Hence $\alpha^*$ falsifies the clause $\bar L[g(u),g(v)]\vee\bar V[g(u),i]\vee D[g(v),i,0]$ of $G^*$, a contradiction. Cases (A16)-(A18) are similar.
\item[--] In case (A19), $C\res \alpha=0$ implies  that $\alpha^*$ falsifies the clause $\bar I[g(u),j]\vee D[g(u),i,b]$ of~$G^*$, a contradiction. Case (A20) is similar.
\item[--] In case (A21), $\alpha(\bar D[s,i,b])=0$ implies $s\in\Dom(g)$ and $\alpha^*$ falsifies the $\bar D[g(s),i,b]$. 
But this is a clause of $G^*$ since $g(s)=t(s)=s^*$ by (H3) -- contradiction.
\end{enumerate}
This finishes the proof of Claim \ref{cl:axtrue}.
\end{proof}

We are ready to finish the proof of the lemma. Let~$P$ be the set of
conditions~$p = (g,h)$ with~$|\Dom(g)|\le w$.  Assume that there
exists a Resolution refutation of~$\REF(F,s)$ of index-width smaller
than~$w$.  Let~$p_0 = (g_0,h_0)$ where~$g_0 = h_0 = \{(0,0)\}$ and
note that~$\partial\Img(g_0) = \emptyset$, so~$p_0 \in P$. The
assignment~$\alpha(p_0)$ is empty and falsifies the empty clause, the
last clause of the refutation. Let~$E$ be the earliest clause in the
refutation such that~$E\res \alpha(p)=0$ for some condition~$p \in P$.
In particular,~$\alpha(p)$ is defined on all variables of~$E$.  By
Claim~\ref{cl:axtrue},~$E$ is not a weakening of a clause from~$G$.
Hence,~$E$ is obtained by a cut of earlier clauses~$C$ and~$D$ on some
variable. Let~$u\in[s]$ be the index mentioned by this variable.
Choose~$p'$ according to Claim~\ref{cl:extend}. Then~$\alpha(p')$ is
defined on all variables in~$C$,~$D$, and~$E$ and extends the partial
assignment~$\alpha(p)$, so falsifies~$E$. By soundness it
falsifies~$C$ or~$D$, say, it falsifies~$C$. Let~$p''$ be the
restriction of~$p'$ to the indices mentioned
in~$C$. Then~$\alpha(p'')$ falsifies~$C$ and~$p''\in P$
by~Claim~\ref{cl:res}. This contradicts the choice of~$E$.
\end{proof}

\begin{remark}
The width lower bound in the previous lemma does not have much to do
with Resolution; a more general version can be formulated using the
notions of semantic refutations and Poizat width from \cite{am}.  The
notion of a Poizat tree is straightforwardly adapted to the
many-sorted structures coding refutations. Define {\em index
  Poizat width} like Poizat width but using the {\em index height} of
a Poizat tree: the maximum over its branches of the number of indices
from~$[s]$ appearing in queries of the branch. Then, the conclusion of
the above lemma can be strengthened to: every semantic refutation
of~$\REF(F,s)$ contains a formula of index Poizat width at
least~$w/3$.
\end{remark}

\section{Relativized formula RREF} \label{sec:rref}

Given a CNF formula~$F$ with~$n$ variables and~$m$ clauses, and a
natural number~$s \geq 1$, we define the CNF formula~$\RREF(F,s)$ as
follows. We again write~$X_1,\ldots,X_n$ for the variables
and~$C_1,\ldots,C_m$ for the clauses of~$F$. The CNF
formula~$\RREF(F,s)$ has the same variables as~$\REF(F,s)$ plus
\begin{itemize} \itemsep=0pt
\item $P[u]$ for $u \in [s]$ indicating that $u$ is an ``active'' index.
\end{itemize}
The clauses of~$\RREF(F,s)$ are very similar to those of~$\REF(F,s)$
with a few additional literals in each clause, and three additional
types of clauses. For easiness of future reference, we explicitly
listed the new set of clauses in Table~\ref{fig:rrefclauses} of
Appendix~\ref{sec:appendix}. In words,~$\RREF(F,s)$ says that the
lines indexed by its at most~$s$ active indices describe a Resolution
refutation of~$F$, and it does not put any restriction on the structure
of the lines on inactive indices. 

First we prove the lower bound:

\begin{lemma} \label{lem:lowerbound}
There is an integer~$n_0 \geq 0$ such that for all integers~$n$
and~$w$ with~$n\ge n_0$ and~$20\le w\le 2^n/(13n)$ and every
unsatisfiable CNF formula~$F$ with~$n$ variables, every Resolution
refutation of~$\RREF(F,13nw)$ has length bigger than~$2^{2w/5}$.
\end{lemma}

\begin{proof}  Let 
$F$ be an unsatisfiable CNF with~$n$ variables and~$m$ clauses
  and~$20\le w\le 2^n/(13n)$. Assume~$\Pi$ is a Resolution refutation
  of~$\RREF(F,t)$ of length~$\ell\le 2^{2w/5}$ where~$t := 13nw$.  We
  derive a contradiction assuming at various places that~$n$ is large
  enough and this determines the constant~$n_0$. It will be clear that
  it does not depend on~$F$ or~$w$.
  
  We define a random restriction~$\rho$ to (a subset of) the variables
  of~$\RREF(F,t)$ by the following random experiment:
\begin{enumerate}\itemsep=0pt
\item independently for every $u \in [t]$, map $P[u]$ to $1$ or $0$
each with probability $1/2$; 
\item let $A$ be the set of $u \in [t]$ for which $P[u]$ is mapped to $1$;
\item for every $u\in A$ and $v\in[t]\setminus A$, 
map both $L[u,v]$ and $R[u,v]$ to 0;
\item independently for every $u\in[t]\setminus A$ and every variable that
  mentions $u$, map the variable to $1$ or $0$ each with
  probability~$1/2$.
\end{enumerate}
A literal that mentions~$u\in[t]$ evaluates to 1 under~$\rho$ with
probability at least~$1/4$, namely in the event that~$P[u]$ is mapped
to~$0$ in step~1 and the right value is chosen in step~4. Thus, the
probability that a clause of index-width at least~$w$ is not satisfied
by~$\rho$ is at most~$(3/4)^{w}$.  By the union bound, the probability
that~$\Pi\res\rho$ contains a clause of index-width at least~$w$ is at
most~$\ell\cdot(3/4)^{w}$, which is strictly less than~$1/4$
for~$\ell\le 2^{2w/5}$ (here we use that~$w\ge 20$).  Note the clauses
of~$\Pi\res\rho$ use variables of~$\REF(F,t)$, so index-width is
well-defined.


The cardinality of the random subset~$A$ is a symmetric binomial
random variable with expectation~$t/2 = 13nw/2$. By the Chernoff bound
there is a real~$\epsilon>0$, independent of~$F$ and~$w$, such
that~$|A| < 6nw$ with probability at most~$2^{-\epsilon nw}$. For
large enough~$n$ this is strictly less than~$1/4$.  Further,~$P[t]$ is
mapped to~$1$ with probability~$1/2$. Thus, for large enough~$n$, by
the union bound, there exists a restriction~$\rho$ in the support of
the above distribution, say, with associated set~$A\subseteq [t]$,
such that:
\begin{enumerate}\itemsep=0pt
\item[(i)] $\Pi\res \rho$ has index-width smaller than $w$; 
\item[(ii)] $|A|\ge 6nw$; 
\item[(iii)] $\rho$ maps $P[t]$ to $1$, so $t \in A$.
\end{enumerate}
By (iii),~$\rho$ satisfies (A24). Also,~$C\res \rho=1$ for~$C$ a
clause of type (A22) or (A23) because this holds for every restriction
in the support of the distribution.

Let~$s = |A|$ and let~$\REF(F,A)$ be defined as~$\REF(F,s)$ except
that we use~$A$ instead~$[s]$ as index set, with~$t$ in the role
of~$s$.  More precisely,~$\REF(F,A)$ is obtained from~$\REF(F,s)$ by a
copy of variables: a variable is replaced by the variable
(of~$\REF(F,t)$) obtained by changing its index~$u\in[s]$
(and~$v\in[s]$) to the~$u$-th (and the~$v$-th) member of~$A$.
 
We claim that for every clause~$C\in \RREF(F,t)$ we either have~$C\res
\rho=1$ or~$C\res \rho\in \REF(F,A)$. We already checked this for
(A22)-(A24) and are left with (A1)-(A21).  For example, if~$C$ is a
clause of type~(A3), then~$C\res\rho=1$ if~$u\notin A$, and
otherwise~$C\res\rho=L[u,0]\vee\bigvee_{v\in A }L[u,v]$ is a clause
in~$\REF(F,A)$. The case that~$C$ is of type (A4) is similar. The
remaining cases are obvious.  Thus,~$\Pi\res\rho$ is a Resolution
refutation of~$ \REF(F,A)$ of length at most~$\ell$.  By~(i) and~(ii),
if~$n$ is large enough, this contradicts
Lemma~\ref{lem:indexwidthlowerbound} (note~$2^n\ge t\ge s$).
  \end{proof}




The next lemma gives a polynomial upper bound on the length of
Resolution refutations of $\RREF(F,s)$ when $F$ is satisfiable. In
fact, its second statement gives an upper bound that is possibly
sublinear in the size of $\RREF(F,s)$. This second statement is not
needed to prove Theorems~\ref{thm:main} and \ref{thm:approx}.

\begin{lemma} \label{lem:upperbound} There is a polynomial $p(s,n,m)$ 
such that for all integers $n,m,s \geq 1$  and every
satisfiable CNF formula~$F$ with~$n$ variables and $m$ clauses, there
  exists a Resolution refutation of $\RREF(F,s)$ of length at most $p(s,n,m)$. In fact, $p(s,n,m)\in O((snm)^2)$.
\end{lemma}

\begin{proof} 
Let $F$ be a satisfiable CNF with variables 
$X_1,\ldots,X_n$ and clauses $C_1,\ldots,C_m$.
Let~$\alpha : \{X_1,\ldots,X_n\} \rightarrow
\Bool$ 
be an assignment that satisfies $F$. 
We derive the clauses
\begin{equation*}
\textit{True}(u)\ :=\ \bar{P}[u] \vee D[u,1,\alpha(X_1)] \vee \cdots \vee D[u,n,\alpha(X_n)]
\end{equation*}
for $u = 1,2,3,\ldots,s$ in order. Then $n$ many cuts with (A21) and
one cut with (A24) yield the empty clause.

First, we derive, for all $u\in[s]$ and $j\in[m]$, as $sm$ many weakenings
of clauses of~$\RREF(F,s)$, the auxiliary clauses
\begin{equation*}\label{eq:IT}
A_0(j,u)\ :=\ \bar I[u,j]\vee\textit{True}(u).
\end{equation*}
Since $\alpha$ satisfies $F$ we can choose for every $j\in[m]$ some
$i_j\in[n]$ such that $X_{i_j}^{(\alpha(X_{i_j}))}$ appears
in~$C_j$. Then $\bar P[u]\vee\bar I[u,j]\vee D[u,i_j,\alpha(X_{i_j})]$
is a clause of $\RREF(F,s)$, namely (A19).
But~$A_0(j,u)$ is a weakening of this.

We derive $\textit{True}(u)$ for $u=1$ through $s+m+2$ many
cuts. Through a sequence of $s$ many cuts, starting at (A4) and using
(A14) for all $v\in[s]$, get $\bar P[u]\vee R[u,0]$. Cut this with~(A12)
to get $\bar P[u]\vee\bar I[u,0]$. Cut this with (A2), followed by a
sequence of $m$ many cuts with all~$A_0(j,u)$ for $j\in[m]$ to get $\textit{True}(u)$.

Now assume $u>1$ and $\textit{True}(v)$ have been derived for all
$v<u$. First, we derive for  every $i\in[n]$ the auxiliary clause 
\begin{equation*}
A_1(i,u):=\left\{
\begin{array}{ll}
\bar V[u,i]\vee   L[u,0]\vee\textit{True}(u)&\textup{if }\alpha(X_i)=1\\
\bar V[u,i]\vee   R[u,0]\vee\textit{True}(u)&\textup{if }\alpha(X_i)=0.
\end{array}\right.
\end{equation*}
We treat the case $\alpha(X_i)=1$, the case $ \alpha(X_i)=0$ is analogous. Let $v\in[u-1]$.
 Cut (A15) with $\bar P[v]\vee \bar
D[v,i,0]\vee\bar D[v,i,1]$ of type (A20) on $D[v,i,0]$ to get
\begin{equation*}
\bar P[u]\vee\bar P[v]\vee\bar L[u,v]\vee \bar V[u,i]\vee \bar D[v,i,\alpha(X_i)].
\end{equation*}
Cut this with $\textit{True}(v)$ on $ D[v,i,\alpha(X_i)]$ to get
\begin{equation*}
\textstyle \bar P[u]\vee\bar P[v]\vee\bar L[u,v]\vee \bar V[u,i]\vee \bigvee_{ i'\in[n]\setminus\{i\}}D[v, i',\alpha(X_{ i'})].
\end{equation*}
Cut this with (A17) on 
$D[v, i',\alpha(X_{ i'})]$ for every $i'\in[n]\setminus\{i\}$, and then with with (A22) on~$P[v]$ to get
\begin{equation}\label{eq:aa1}
\textstyle \bar P[u]\vee\bar L[u,v]\vee \bar V[u,i]\vee \bigvee_{ i'\in[n]\setminus\{i\}}D[u, i',\alpha(X_{ i'})].
\end{equation}
Now cut (A3) with this formula for all $v\in[u-1]$,
and with (A13) for all $u\le v\le s$ to get the following subclause of
 $A_1(i,u)$
\begin{equation*}
\textstyle \bar P[u]\vee L[u,0]\vee \bar V[u,i]\vee \bigvee_{ i'\in[n]\setminus\{i\}}D[u, i',\alpha(X_{ i'})].
\end{equation*}
For every $v\in[u-1]$, the clause \eqref{eq:aa1} is derived with $n+2$
cuts. Thus, $A_1(i,u)$ is derived with~$(n+2)(u-1)+s$ many cuts. Doing this for all $i\in[n]$ amounts to 
$n(n+2)(u-1)+ns$ many cuts.

Having derived the auxiliary clauses $A_1(i,u)$ we now derive
$\textit{True}(u)$ in a sequence of~$n+m+4$ cuts. In a sequence of
$n$ many cuts, cut (A1) with $A_1(i,u)$ for all $i\in[n]$, to get
\begin{equation*}
V[u,0]\vee  L[u,0]\vee  R[u,0]\vee\textit{True}(u).
\end{equation*}

Cut with (A9) on $V[u,0]$, then with (A11) on $L[u,0]$, and then with
(A12) on $R[u,0]$ to get
$
\bar I[u,0]\vee \textit{True}(u).
$
Cut (A2) with this and then with $A_0(j,u)$ for all $j\in[m]$ in
sequence to get $ \textit{True}(u)$ as desired.

In total, the refutation uses
\begin{equation*}\label{eq:cuts}
\textstyle
 (s+2+m)+ (n+1) + \sum_{u=2}^s \big( n(n+2)(u-1)+ns + (n+m+4)\big) 
\end{equation*}
many cuts: the first term counts the cuts in the derivation of
$\textit{True}(1)$, the second term counts the cuts to get the empty clause from~$\textit{True}(s)$, and
each term in the big sum counts the cuts in the
derivation of $\textit{True}(u)$ for $u = 2,\ldots,s$.
The length of the refutation is bounded by the number of cuts plus the
$sm$ weakenings to get the $A_0(j,u)$'s plus the number of clauses of $\RREF(F,s)$. But, 
in fact, the clauses (A7) and (A8) are not used by the given
refutation, and $\REF(F,s)$ has at most
$O((snm)^2)$ many other clauses.
%
\end{proof}

\begin{remark} 
In the proof of the upper bound Lemma~\ref{lem:upperbound}, the
built-in linear order in the definition of $\RREF$ plays a crucial
role. This refers to the side conditions $u \leq v$ in clauses~(A13)
and~(A14) of the definition of RREF in Table~\ref{fig:rrefclauses} of
Appendix~\ref{sec:appendix}. Indeed, it is not hard see that if the
linear order were not built-in but interpreted, through new
propositional variables~$O[u,v]$ and its corresponding clause axioms,
then the resulting version of $\RREF(F,s)$ would be exponentially hard
for resolution independently of the satisfiability or unsatisfiability
of $F$. This follows from an ``infinite model argument'' similar to
the proof of the main theorem in~\cite{dr}.
\end{remark}

\section{Proofs of the hardness results} \label{sec:hardness}


%
%


In this section we derive Theorems~\ref{thm:main} and \ref{thm:approx}
stated in the Introduction.

\begin{proof}[Proof of Theorem~\ref{thm:approx}]
It suffices to define
  $G$ on 3-CNF formulas $F$ with a sufficiently large number of
  variables $n$. Note $m\le 8n^3$ for $m$ the number of clauses of
  $F$. We set
\begin{equation*}
G(F):=\RREF(F,13n^2).
\end{equation*}
Note  $G(F)$
has size between $n^{1/q}$ and $n^q$ for some constant $q >
0$. Thus,~(a) follows from the first statement of
Lemma~\ref{lem:upperbound} for some constant $c > 0$, and (b) follows
from Lemma~\ref{lem:lowerbound} for~$w:=n$ and some constant $d > 0$
(note that $20\le w\le 2^n/(13n)$ for sufficiently large~$n$).
\end{proof}

Our main result, Theorem~\ref{thm:main}, is implied by the more
general statement below.  We say that Resolution is {\em automatizable
  in time~$t$} if there is an algorithm that, given an
unsatisfiable~CNF formula~$F$,
computes some Resolution refutation of~$F$ in
time~$t(r(F)+s(F))$.  Recall that a
function~$t : \nats \rightarrow \nats$ is {\em time-constructible} if
there is an algorithm that given~$1^n$ (the string of~$n$ many 1's)
computes~$1^{t(n)}$ in
time~$O(t(n))$. We say that~$t$ is {\em subexponential} if~$t(n)\le
2^{n^{o(1)}}~$.

\begin{theorem}
  Let $t : \nats \rightarrow \nats$ be time-constructible,
  non-decreasing and subexponential. If Resolution is automatizable in
  time $t$, then there are polynomials $q(n)$ and $r(n)$ and an
  algorithm that, given a 3-CNF formula $F$ with $n$ variables,
  decides in time~$O(t(r(n))+q(n))$ whether~$F$ is satisfiable.
\end{theorem}

\begin{proof}
  Assume that Resolution is automatizable in time $t$ and choose
  $c$, $d$ and $G$ from Theorem~\ref{thm:approx}. Let $q$ be as in the
  proof given above, so $G(F)$ has size at most $n^q$ for every 3-CNF
  formula~$F$ with a sufficiently large number $n$ of variables.

  Consider the following algorithm. Given a 3-CNF formula $F$ with $n$
  variables, compute the formula $G(F)$ and run the automating
  algorithm for up to $t(n^q+n^{qc})$ steps. If the algorithm returns
  a Resolution refutation within the allotted time, then output
  `satisfiable'. Else output `unsatisfiable'.

  It is clear that the algorithm runs in time
  $O(t(n^{q}+n^{qc})+q(n))$ for some polynomial~$q(n)$; here we use
  that $t$ is time-constructible.  It suffices to show that it is
  correct on 3-CNF formulas $F$ with a sufficiently large number of
  variables $n$.
  If $F$ is satisfiable, then by (a) of Theorem~\ref{thm:approx},
  $G(F)$ has a Resolution refutation of length at most $n^{qc}$, so
  the automating algorithm computes a refutation within the alloted
  time and we answer `satisfiable'; here we use that $t$ is
  non-decreasing.  If $F$ is unsatisfiable, then by (b) of
  Theorem~\ref{thm:approx}, no Resolution refutation of $G(F)$ has
  length at most $O(t(n^{q}+n^{qc}))$, so the automating algorithm
  cannot compute one within the allotted time and we answer
  `unsatisfiable'; here we use that $t$ is subexponential.
\end{proof}

\section{Concluding remarks} \label{sec:conclusions}

This final section contains the observation that, by a
padding argument, the constants~$c$
and~$d$ in Theorem~\ref{thm:approx} can be chosen arbitrarily close to
$1$ and $2$, respectively, and finishes with some questions.


\begin{proof}[Proof of Theorem~\ref{thm:approx} for $c=1+\epsilon$ and $d=2+\epsilon$]
Given $\epsilon>0$ we define $G(F,t)$ for a 3-CNF~$F$ and a natural
$t>0$ and verify~(a) and~(b) for $G(F):=G(F,t)$ assuming that~$t$ and
$n$ are sufficiently large; again, $n$ denotes the number of variables
of $F$. The meaning of ``sufficiently large'' for $t$ will depend only
on~$\epsilon$.  Write $w:= n^t$ and let~$G(F,t):=\RREF'(F,13nw)$ be
obtained from $\RREF(F,13nw)$ by deleting the clauses of type (A7)
and~(A8). As has been noted in the proof, Lemma~\ref{lem:upperbound}
holds true for $\RREF'(F,13nw)$ instead $\RREF(F,13nw)$. Since~$F$ has
at most $ 8n^3$ clauses, the size~$r(t)$ of $G(F,t)$ satisfies~
$n^{2t}\le r(t)\le n^{2t+c_0} $ for some constant $c_0 > 0$, a number
independent of~$F$ and~$t$.  By Lemma~\ref{lem:upperbound},~$s(G(F,t))
< n^{2t+c_1}$ for some constant $c_1 > 0$; but this is at most
$r(t)^{1+\epsilon}$ if $t > c_1/2\epsilon$.  By
Lemma~\ref{lem:lowerbound}, $s(G(F,t))>2^{2n^t/5}$, but this is more
than~$2^{r(t)^{1/(2+\epsilon)}}$ if~$t> c_0/\epsilon$.
\end{proof}

The reduction above falls short to rule out weak automatizability of
Resolution. For this we would need that $s(G(F))=\infty$ when $F$ is
unsatisfiable, i.e., that $G(F)$ is satisfiable, but this is unlikely
to hold for a polynomial time $G$ as it would put 3-SAT in co-NP. We
refer to \cite{ab} for a proof of equivalence of the different
characterizations of weak automatizability used here and in the
Introduction. The main problem left open by the current work is to
find more convincing evidence that Resolution is not weakly
automatizable.

On the more technical side, we would like to know whether the
formulas~$\REF(F,p(n))$ are hard for Resolution where~$F$ ranges over
unsatisfiable CNF formulas with~$n$ variables and~$p$ is some fixed
polynomial. We conjecture that this is the case but we only succeeded
in establishing a width lower bound\footnote{Recently, M. Garl\'ik
  confirmed the conjecture~\cite{garlik}.}.  Of course, one can define
analogous formulas~$P\Ref(F,s)$ for any proof system~$P$.  For all we
know it could be that such formulas~$P\Ref(F,p(n))$ are hard for
strong proof systems~$P$ like Frege or Extended Frege. Of course, it
would be a major breakthrough to prove this, even under some plausible
computational hardness hypothesis. We refer to~\cite[Chapter
  27]{kraforce} for a discussion.

\appendix

\section{Formulas REF and RREF} \label{sec:appendix}

In this appendix we include the detailed lists of clauses of the
formulas $\REF$ and $\RREF$. Recall, we use bars to denote the negation of
the variables, e.g., $\bar L[u,v]$ denotes the negation of~$L[u,v]$.

{\small
\begin{longtable}{rlll}
(A1) & $V[u,0] \vee V[u,1] \vee \cdots \vee V[u,n]$ & & $u \in [s]$, \\
(A2) & $I[u,0] \vee I[u,1] \vee \cdots \vee I[u,m]$ & & $u \in [s]$, \\
(A3) & $L[u,0] \vee L[u,1] \vee \cdots \vee L[u,s]$ & & $u \in [s]$, \\
(A4) & $R[u,0] \vee R[u,1] \vee \cdots \vee R[u,s]$ & & $u \in [s]$, \\
(A5) & $\bar{V}[u,i] \vee \bar{V}[u,i']$ & & $u \in [s]$, $i,i'
  \in [n] \cup \{0\}$, $i \not= i'$, \\
(A6) & $\bar{I}[u,j] \vee \bar{I}[u,j']$ & & $u \in [s]$, $j,j' \in [m] \cup \{0\}$, $j \not= j'$, \\
(A7) & $\bar{L}[u,v] \vee \bar{L}[u,v']$ & & $u \in [s]$, $v,v' \in [s] \cup \{0\}$, $v \not= v'$, \\
(A8) & $\bar{R}[u,v] \vee \bar{R}[u,v']$ & & $u \in [s]$, $v,v' \in [s] \cup \{0\}$, $v \not= v'$, \\
(A9) & $\bar{I}[u,0] \vee \bar{V}[u,0]$ & & $u \in [s]$, \\
(A10)& $I[u,0] \vee V[u,0]$ & & $u \in [s]$, \\
(A11) & $\bar{I}[u,0] \vee \bar{L}[u,0]$ & & $u \in [s]$, \\
(A12) & $\bar{I}[u,0] \vee \bar{R}[u,0]$ & & $u \in [s]$, \\
(A13) & $\bar{L}[u,v]$ & & $u,v \in [s]$, $u \leq v$, \\
(A14) & $\bar{R}[u,v]$ & & $u,v \in [s]$, $u \leq v$, \\
(A15) & $\bar{L}[u,v] \vee \bar{V}[u,i] \vee D[v,i,0]$ & & $u,v \in [s]$,
$i \in [n]$, \\
(A16) & $\bar{R}[u,v] \vee \bar{V}[u,i] \vee D[v,i,1]$ & & $u,v \in [s]$,
$i \in [n]$, \\
(A17) & $\bar{L}[u,v] \vee \bar{V}[u,i] \vee \bar{D}[v,i',b] \vee D[u,i',b]$ & &
$u,v \in [s]$, $i,i' \in [n]$, $b \in \Bool$, $i' \not= i$, \\
(A18) & $\bar{R}[u,v] \vee \bar{V}[u,i] \vee \bar{D}[v,i',b] \vee D[u,i',b]$ & &
$u,v \in [s]$, $i,i' \in [n]$, $b \in \Bool$, $i' \not= i$, \\
(A19) & $\bar{I}[u,j] \vee D[u,i,b]$ & & $u \in [s]$, $j \in [m]$,
$X_i^{(b)} \in C_j$, \\
(A20) & $\bar D[u,i,0] \vee \bar D[u,i,1]$ & & $u \in [s]$, $i \in [n]$,\\
(A21) & $\bar{D}[s,i,b]$ & & $i \in [n]$, $b \in \Bool$. \\
& & & \\
\caption{Clauses of $\REF$.}
\label{fig:refclauses}
\end{longtable}
}
\noindent In the clauses of $\REF$, clauses (A1)-(A8) say $V$, $I$,
$L$ and $R$ are functions with the appropriate domains and ranges,
(A9)-(A10) express (R1), (A11)-(A12) express (R2), (A13)-(A14) express
(R3), (A15)-(A16) express (R4), (A17)-(A18) express (R5), (A19)
expresses (R6), (A20) expresses (R7), and (A21) expresses (R8).

\bigskip
{\small
\begin{longtable}{rlll}
(A1) & $\bar{P}[u] \vee V[u,0] \vee V[u,1] \vee \cdots \vee V[u,n]$ & & $u \in [s]$, \\
(A2) & $\bar{P}[u] \vee I[u,0] \vee I[u,1] \vee \cdots \vee I[u,m]$ & & $u \in [s]$, \\
(A3) & $\bar{P}[u] \vee L[u,0] \vee L[u,1] \vee \cdots \vee L[u,s]$ & & $u \in [s]$, \\
(A4) & $\bar{P}[u] \vee R[u,0] \vee R[u,1] \vee \cdots \vee R[u,s]$ & & $u \in [s]$, \\
(A5) & $\bar{P}[u] \vee \bar{V}[u,i] \vee \bar{V}[u,i']$ & & $u \in [s]$, $i,i'
  \in [n] \cup \{0\}$, $i \not= i'$, \\
(A6) & $\bar{P}[u] \vee \bar{I}[u,j] \vee \bar{I}[u,j']$ & & $u \in [s]$, $j,j' \in [m] \cup \{0\}$, $j \not= j'$, \\
(A7) & $\bar{P}[u] \vee \bar{L}[u,v] \vee \bar{L}[u,v']$ & & $u \in [s]$, $v,v' \in [s] \cup \{0\}$, $v \not= v'$, \\
(A8) & $\bar{P}[u] \vee \bar{R}[u,v] \vee \bar{R}[u,v']$ & & $u \in [s]$, $v,v' \in [s] \cup \{0\}$, $v \not= v'$, \\
(A9) & $\bar{P}[u] \vee \bar{I}[u,0] \vee \bar{V}[u,0]$ & & $u \in [s]$, \\
(A10)& $\bar{P}[u] \vee I[u,0] \vee V[u,0]$ & & $u \in [s]$, \\
(A11) & $\bar{P}[u] \vee \bar{I}[u,0] \vee \bar{L}[u,0]$ & & $u \in [s]$, \\
(A12) & $\bar{P}[u] \vee \bar{I}[u,0] \vee \bar{R}[u,0]$ & & $u \in [s]$, \\
(A13) & $\bar{P}[u] \vee \bar{L}[u,v]$ & & $u,v \in [s]$, $u \leq v$, \\
(A14) & $\bar{P}[u] \vee \bar{R}[u,v]$ & & $u,v \in [s]$, $u \leq v$, \\
(A15) & $\bar{P}[u] \vee \bar{P}[v] \vee \bar{L}[u,v] \vee \bar{V}[u,i] \vee D[v,i,0]$ & & $u,v \in [s]$,
$i \in [n]$, \\
(A16) & $\bar{P}[u] \vee \bar{P}[v] \vee \bar{R}[u,v] \vee \bar{V}[u,i] \vee D[v,i,1]$ & & $u,v \in [s]$,
$i \in [n]$, \\
(A17) & $\bar{P}[u] \vee \bar{P}[v] \vee \bar{L}[u,v] \vee \bar{V}[u,i] \vee \bar{D}[v,i',b] \vee D[u,i',b]$ & &
$u,v \in [s]$, $i,i' \in [n]$, $b \in \Bool$, $i' \not= i$, \\
(A18) & $\bar{P}[u] \vee \bar{P}[v] \vee \bar{R}[u,v] \vee \bar{V}[u,i] \vee \bar{D}[v,i',b] \vee D[u,i',b]$ & &
$u,v \in [s]$, $i,i' \in [n]$, $b \in \Bool$, $i' \not= i$, \\
(A19) & $\bar{P}[u] \vee \bar{I}[u,j] \vee D[u,i,b]$ & & $u \in [s]$, $j \in [m]$,
$X_i^{(b)} \in C_j$, \\
(A20) & $\bar{P}[u] \vee \bar D[u,i,0] \vee \bar D[u,i,1]$ & & $u \in [s]$, $i \in [n]$,\\
(A21) & $\bar{P}[s] \vee \bar{D}[s,i,b]$ & & $i \in [n]$, $b \in \Bool$, \\
(A22) & $\bar{P}[u] \vee\bar{L}[u,v] \vee P[v]$ & & $u \in [s]$, $v \in [s]$, \\
(A23) & $\bar{P}[u] \vee\bar{R}[u,v] \vee P[v]$ & & $u \in [s]$, $v \in [s]$, \\
(A24) & $P[s]$. & & \\
& & & \\
\caption{Clauses of RREF}
\label{fig:rrefclauses}
\end{longtable}
}
\noindent The clauses of $\RREF$ are the same as for $\REF$
  but we add to each clause the literals $\bar{P}[u]$ with $u\in[s]$
  mentioned by the clause. More precisely, $\bar P[u]$ is added to the
  clauses (A1)-(A14) and (A19) and (A20), both $\bar P[u]$ and $\bar P[v]$
  are added to the clauses (A15)-(A18), and $\bar P[s]$ is added to
  clause~(A21). Further, we add three additional types of clauses
  numbered (A22)-(A24).

\subsection*{Acknowledgments}
Both authors were partially funded by European Research Council (ERC)
under the European Union's Horizon 2020 research and innovation
programme, grant agreement ERC-2014-CoG 648276 (AUTAR). First author
partially funded by MICCIN grant TIN2016-76573-C2-1P (TASSAT3). We are
grateful to Ilario Bonacina and Michal Garlik for their very useful
comments on an earlier draft of this paper.

\end{document}